\documentclass[a4paper,11pt]{article}
\pdfoutput=1

\usepackage[english]{babel}
\usepackage[a4paper,tmargin=3truecm,bmargin=3truecm,rmargin=2.5truecm,
lmargin=2.5truecm,twoside,verbose=true]{geometry}

\usepackage{cancel,graphicx}
\usepackage{amsmath,amssymb}
\usepackage[amsmath, hyperref, thmmarks]{ntheorem}
\usepackage[all]{xypic}

\usepackage[pdftex]{hyperref}

\numberwithin{equation}{section}

\allowdisplaybreaks[1]

\newcommand\caA{{\mathcal A}}

\newcommand\caG{{\mathcal G}}
\newcommand\caL{{\mathcal L}}
\newcommand\caM{{\mathcal M_\Sigma}}
\newcommand\caMst{{\mathcal M_{\text{st}}}}
\newcommand\caS{{\mathcal S}}

\newcommand\wx{{\widetilde x}}
\newcommand\gone{{ \mathchoice {1\mskip-4mu\mathrm{l} } {1\mskip-4mu\mathrm{l} }{1\mskip-4.5mu\mathrm{l} } {1\mskip-5mu\mathrm{l}} }}
\newcommand\gR{{\mathbb R}}

\newcommand\fois{\mathord{\cdot}}
\newcommand\dd{{\text{\textup{d}}}}

\theoremsymbol{}
\theorembodyfont{\slshape}
\theoremheaderfont{\normalfont\bfseries}
\theoremseparator{}
\newtheorem{Theorem}{Theorem}[section]
\newtheorem{theorem}[Theorem]{Theorem}

\newtheorem{proposition}[Theorem]{Proposition}

\newtheorem{lemma}[Theorem]{Lemma}

\newtheorem{corollary}[Theorem]{Corollary}

\theorembodyfont{\upshape}
\theoremsymbol{\ensuremath{\blacklozenge}}

\newtheorem{example}[Theorem]{Example}

\newtheorem{definition}[Theorem]{Definition}

\theoremstyle{nonumberplain}
\theoremheaderfont{\scshape}
\theorembodyfont{\normalfont}
\theoremsymbol{\ensuremath{\blacksquare}}

\newtheorem{proof}{Proof}
\qedsymbol{\ensuremath{_\blacksquare}}
\theoremclass{LaTeX}

\title{Symmetries of noncommutative scalar field theory\footnote{Work
supported by the Belgian Interuniversity Attraction Pole (IAP) within the framework ``Nonlinear systems, stochastic processes, and statistical mechanics'' (NOSY).}}
\author{Axel de Goursac$^{a}$, Jean-Christophe Wallet$^b$}
\date{}

\begin{document}

\maketitle
%\vspace*{-1cm}
\begin{center}
\textit{$^a$D\'epartement de Math\'ematiques et Physique, \\ Universit\'e Catholique de Louvain,\\ Chemin du Cyclotron, 2,\\
1348 Louvain-la-Neuve, Belgium\\
  e-mail: \texttt{axelmg@melix.net}}\\
\textit{$^b$Laboratoire de Physique Th\'eorique, B\^at.\ 210\\
    Universit\'e Paris XI,  F-91405 Orsay Cedex, France\\
    e-mail: \texttt{jean-christophe.wallet@th.u-psud.fr}}\\[1ex]
\end{center}%

\vskip 2cm

\begin{abstract}
We investigate symmetries of the scalar field theory with harmonic term on the Moyal space with euclidean scalar product and general symplectic form. The classical action is invariant under the orthogonal group if this group acts also on the symplectic structure. We find that the invariance under the orthogonal group can be restored also at the quantum level by restricting the symplectic structures to a particular orbit.
\end{abstract}

\vfill
%\begin{flushleft}
%LPT-Orsay/
%\end{flushleft}

\pagebreak

\section{Introduction}

In the past few years, there has been a growing interest in the noncommutative quantum fields theories (for a review, see \cite{Wulkenhaar:2006si}). These theories on ``spaces'' coming from noncommutative geometry \cite{Connes:1994} are indeed strong candidates for new physics behind the Standard Model of particle physics. Moreover, fields theories defined on the Moyal space \cite{Bayen:1978}, one of the simplest example of noncommutative space, can be seen as an effective regime of string theory \cite{Seiberg:1999vs} and matrix theory \cite{Connes:1997cr}.

The simplest generalization of the $\varphi^4$ commutative scalar theory gives rise to a new type of divergence, the ultraviolet-infrared (UV/IR) mixing \cite{Minwalla:1999px}, which is responsible for the non-renormalizability of this model. The first solution to this problem of UV/IR mixing, due to Grosse and Wulkenhaar, was to add some harmonic term in the action \cite{Grosse:2004yu}, and the resulting theory is then renormalizable up to all orders in perturbation \cite{Grosse:2004yu,Gurau:2005gd,Gurau:2007fy}. The renormalizability of this model seems to be related to a new symmetry thanks to this harmonic term, the Langmann-Szabo duality \cite{Langmann:2002cc}, which exchanges positions and impulsions at the level of the quadratic terms of the action. This Langmann-Szabo duality has also been interpreted in the framework of superalgebras as a grading symmetry \cite{deGoursac:2008bd,deGoursac:2010zb} and an adapted differential calculus has been exhibited (see also \cite{Cagnache:2008tz,Wallet:2008bq}). Note that another interpretation of the harmonic term has been given in \cite{Buric:2009ss}. Moreover, the vacuum solutions of this theory have been studied in \cite{deGoursac:2007uv}, and it has been proved that the beta function vanishes up to irrelevant terms \cite{Grosse:2004by,Disertori:2006nq}. Associated to this scalar field theory with harmonic term, a gauge invariant action, candidate to renormalizability, has been exhibited in \cite{deGoursac:2007gq} (see also \cite{Grosse:2007dm,Wallet:2007em}), but its vacua are always non-trivial \cite{deGoursac:2008rb}.

However, since a symplectic structure $\Sigma$ is necessary to define the Moyal product, the group of rotations is no longer a symmetry group for this theory on the Moyal space as soon as the dimension $D>2$. The question of rotational invariance of the theory is important since it corresponds to the Lorentz invariance, necessary for a physical theory, in the Minkowskian framework. This problem is related to the fact that a symplectic structure is not natural on a configuration (or position) space, contrary to a phase space. A simple standard way to restore the rotational invariance of the classical action is to consider a family of actions labeled by the symplectic structures rather than a single action corresponding to a fixed symplectic structure. Then, one also allows the rotations to act on these symplectic structures. This idea has been used within the Minkowskian framework in \cite{Doplicher:1994tu} (see also \cite{Grosse:2007vr}). In this paper, we find that if we choose a particular orbit and the corresponding family of actions, the correlation functions are the same for each action of this family, so that the quantum action is invariant under the whole orthogonal group (which contains the rotations).

The paper is organized as follows. In subsection 2.1, we introduce the Moyal algebra $\caM$, endowed with a general scalar product $G$ and a general symplectic structure $\Sigma$, and present the standard Moyal space $\mathcal M_{\text{st}}$, on which is defined the Grosse-Wulkenhaar model. Then, we define in subsection 2.2 the classical action of the scalar field theory with harmonic term on $\caM$ and see that it is invariant under the action of the orthogonal group if this group acts also on $\Sigma$. Subsection 2.3 is devoted to recall the notion of symplectic structures adapted to $G$, which form a single orbit of the action of the orthogonal group. By assuming that $\Sigma$ belongs to this particular orbit, we find in subsection 3.1 the expression of the propagator, we express the amplitudes of Feynman graphs in terms of those of the standard theory, and show the invariance of the bare (regularized) effective action under the orthogonal group. As a consequence, we obtain in subsection 3.2 the renormalizability of this model in $D=4$, and the invariance of the renormalized effective action under the orthogonal group. Finally, subsection 3.3 is devoted to the discussion of these results.

\section{ The orthogonal group and symplectic structures on the Moyal space}

\subsection{Presentation of the Moyal space}

We give here the definition of the Moyal space \cite{Bayen:1978} for any scalar product and symplectic structure, and some of its properties. Let $\caS\equiv\caS(\gR^D)$ with even dimension $D$ and $\caS^\prime$ be respectively the space of complex-valued Schwartz functions on $\gR^D$ and its topological dual, the space of tempered distributions on $\gR^D$. Consider a fixed symplectic structure on $\gR^D$, represented by an invertible skew-symmetric real $D\times D$ matrix $\Sigma$, and $\Theta=\theta \Sigma$, where the positive parameter $\theta$ has mass dimension $-2$. Let also $G$ be a scalar product, namely an invertible symmetric real $D\times D$ matrix.
\begin{definition}
The Moyal product associated to $\Theta$ (and $G$) can be conveniently defined as $\star_\Sigma:\caS\times\caS\to\caS$ by: $\forall a,b\in\caS$,
\begin{equation}
(a\star_\Sigma b)(x)=\frac{|\text{det}G|^2}{(\pi\theta)^D|\text{det}\Sigma|}\int d^Dyd^Dz\ a(x+y)b(x+z)e^{-iy\wedge z}, \label{eq:moyal}
\end{equation}
where $y\wedge z=2y_\mu G_{\mu\nu}\Theta^{-1}_{\nu\rho}G_{\rho\sigma}z_\sigma$.
\end{definition}

\begin{proposition}
\begin{itemize}
\item The product \eqref{eq:moyal} is associative, and one has: $\forall a,b,c\in\caS$, $\forall x\in\gR^D$,
\begin{equation}
(a\star_\Sigma b\star_\Sigma c)(x)=\frac{|\det(G)|^2}{(\pi\theta)^D|\det(\Sigma)|}\int d^Dy d^Dz\ a(x+y)b(x+z)c(x-y+z)e^{-iy\wedge z}.
\end{equation}
\item It also satisfies to the tracial identity: $\forall a,b\in\caS$,
\begin{equation}
\int d^Dx\ (a\star_\Sigma b)(x)=\int d^Dx\ a(x) b(x).\label{tracident}
\end{equation}
\end{itemize}
\end{proposition}
\begin{proof}
It is a straightforward computation using the equation
\begin{equation}
\int d^Dx\ e^{-ix\wedge y}=\frac{(\pi\theta)^D|\det(\Sigma)|}{|\det(G)|^2}\delta^D(y).
\end{equation}
\end{proof}

The product \eqref{eq:moyal} can be further extended to $\caS^\prime\times\caS$ (and $\caS \times \caS^\prime$) upon using duality of linear spaces: $\langle T\star_\Sigma a,b \rangle = \langle T,a\star_\Sigma b\rangle$, $\forall T\in\caS^\prime$, $\forall a,b\in\caS$. The Moyal algebra $\caM$ is then defined as
\begin{equation}
\caM=\mathcal L\cap \mathcal R,\label{eq:Moyal-alg}
\end{equation}
where $\mathcal L$ (resp. $\mathcal R$) is the subspace of $\caS^\prime$ whose multiplication from right (resp. left) by any Schwartz function is a subspace of $\caS$.

$(\caM,\star_\Sigma,{}^\dag)$ is a unital involutive algebra which involves in particular the ``coordinate'' functions $x_\mu$ satisfying to the relation defined on $\caM$:
\begin{equation}
[x_\mu,x_\nu]_{\Sigma} =i\theta (G^{-1}\Sigma G^{-1})_{\mu\nu},\label{eq:comrelation}
\end{equation}
where we set $[a,b]_\Sigma=a\star_\Sigma b-b\star_\Sigma a$ and $\{a,b\}_\Sigma=a\star_\Sigma b+b\star_\Sigma a$.

Let us give some relevant properties of the Moyal space. If $\wx_\mu=2G_{\mu\nu}\Theta^{-1}_{\nu\rho}G_{\rho\sigma}x_\sigma$, $\forall a,b\in\caM$,
\begin{subequations}
\begin{align}
\partial_\mu(a\star_\Sigma b)&=(\partial_\mu a)\star_\Sigma b+a\star_\Sigma(\partial_\mu b), \label{eq:relat1}\\
[\wx_\mu,a]_\Sigma&=2i\partial_\mu a,\label{eq:relat2}\\
\{\wx_\mu,a\}_\Sigma&=2\wx_\mu a.\label{eq:relat3}
\end{align}
\end{subequations}
\medskip

The standard $D$-dimensional Moyal space $\caMst$ (see \cite{GraciaBondia:1987kw,Varilly:1988jk} for more details) corresponds to the euclidean metric $G^{\text{st}}$ and standard symplectic structure $\Sigma^{\text{st}}$, which take the following form in four dimensions:
\begin{align}
G^{\text{st}}_{\mu\nu}=\begin{pmatrix} 1 & 0 & 0 & 0 \\ 0 & 1 & 0 & 0 \\ 0 & 0 & 1 & 0 \\ 0 & 0 & 0 & 1 \end{pmatrix}, \quad\Sigma^{\text{st}}_{\mu\nu}=\begin{pmatrix} 0 & -1 & 0 & 0 \\ 1 & 0 & 0 & 0 \\ 0 & 0 & 0 & -1 \\ 0 & 0 & 1 & 0 \end{pmatrix}.\label{stmetric}
 \end{align}
On this standard Moyal space, the Grosse-Wulkenhaar model, which cures the problem of UV/IR mixing is defined by the action \cite{Grosse:2004yu}:
\begin{equation}
S(\phi)=\int d^Dx\Big(\frac 12\partial_\mu\phi\star\partial_\mu\phi +\frac{\Omega^2}{2}(\wx_\mu\phi)\star(\wx_\mu\phi) +\frac{m^2}{2}\phi\star\phi +\lambda\phi\star\phi\star\phi\star\phi\Big),\label{actscal}
\end{equation}
and it is renormalizable up to all orders in perturbation \cite{Grosse:2004yu,Gurau:2005gd} in two and four dimensions. The propagator associated to the action \eqref{actscal} is given by \cite{Gurau:2005qm} ($\widetilde \Omega=\frac{2\Omega}{\theta}$):
\begin{align}
C^{\text{st}}(x,y)&=\frac{\theta}{4\Omega}\left(\frac{\Omega}{\pi\theta}\right)^{\frac D2}\int_0^\infty \!\! \frac{d\alpha}{\sinh^{\frac D2}(\alpha)} e^{-\frac{m^2\alpha}{2\widetilde \Omega}}C^{\text{st}}(x,y,\alpha),\\
C^{\text{st}}(x,y,\alpha) &= \exp\Big(-\frac{\widetilde \Omega}{4}\coth(\frac \alpha 2)(x-y)^2 -\frac{\widetilde \Omega}{4}\tanh(\frac \alpha 2)(x+y)^2\Big).
\end{align}
The associated gauge theory is given by \cite{deGoursac:2007gq}:
\begin{equation}
S=\int d^4x \Big(\frac{\alpha}{4g^2}F_{\mu\nu}\star F_{\mu\nu} +\frac{\Omega'}{4g^2}\{\caA_\mu,\caA_\nu\}^2_\star
+\frac{\kappa}{2} \caA_\mu\star\caA_\mu\Big),\label{actgauge}
\end{equation}
where $F_{\mu\nu}=\partial_\mu A_\nu-\partial_\nu A_\mu-i[A_\mu,A_\nu]_\star$, and $\caA_\mu=A_\mu+\frac 12\wx_\mu$ is the covariant coordinate. It is candidate to renormalizability.

\subsection{Invariance of the classical action under the orthogonal group}

In the notations of the previous subsection, we consider the Moyal algebra $\caM$, for a general scalar product $G$ and a symplectic structure $\Sigma$.

\begin{definition}
 The classical action of the scalar theory with harmonic term on the Moyal space $\caM$ is given by:
\begin{equation}
S(\phi)=\int d^Dx\ \Big(\frac 12G_{\mu\nu}^{-1}(\partial_\mu\phi)(\partial_\nu\phi) +\frac{\Omega^2}{2}G_{\mu\nu}^{-1}\wx_\mu\wx_\nu\phi^2+\frac{m^2}{2}\phi^2+\lambda\phi\star_\Sigma\phi\star_\Sigma\phi\star_\Sigma\phi\Big).\label{actscalcurv}
\end{equation}
\end{definition}
It turns out that this action is covariant under Langmann-Szabo duality \cite{Langmann:2002cc}, which at the level of quadratic terms, basically exchanges $\partial_\mu$ and $\wx_\mu$.

The symmetries associated to $G$ forms the orthogonal group:
\begin{equation}
O(\gR^D,G)=\{\Lambda\in GL(\gR^D),\quad \Lambda^T G\Lambda=G\},
\end{equation}
while those associated to $\Sigma$ are the symplectic isomorphisms:
\begin{equation}
Sp(\gR^D,\Sigma)=\{\Lambda\in GL(\gR^D),\quad \Lambda^T \Sigma \Lambda=\Sigma\}.
\end{equation}
Then, for fixed $G$ and $\Sigma$, the symmetries of the action \eqref{actscalcurv} are orthogonal and symplectic isomorphisms
\begin{equation}
Sym(\gR^D,G,\Sigma)=O(\gR^D,G)\cap Sp(\gR^D,\Sigma),
\end{equation}
whose group is isomorphic to $SU(\frac D2)$.

In order to recover all elements of $O(\gR^D,G)$ as symmetries of the classical action \eqref{actscalcurv}, one can consider $\Sigma$ as a tensor, namely also transforming under the orthogonal group, by the following left action:
\begin{align}
O(\gR^D,G)\times Sympl(\gR^D)&\to Sympl(\gR^D)\nonumber\\
(\Lambda,\Sigma)&\mapsto (\Lambda^{-1})^T\Sigma\Lambda^{-1}\label{actrot}
\end{align}
where $Sympl(\gR^D)$ denotes the set of all symplectic structures on $\gR^D$. The orbits and isotropy groups of this action are then defined by:
\begin{align}
\mathcal O_{\Sigma}&=\{(\Lambda^{-1})^T\Sigma\Lambda^{-1},\quad \Lambda\in O(\gR^D,G)\},\\
Gr_\Sigma&=\{\Lambda\in O(\gR^D,G),\quad (\Lambda^{-1})^T\Sigma\Lambda^{-1}=\Sigma\}=Sym(\gR^D,G,\Sigma).
\end{align}

This is a well-known procedure in field theory and one obtains an invariant classical action.
\begin{proposition}
The classical action \eqref{actscalcurv} is invariant under the following action of the orthogonal group on the field $\phi$ and on $\Sigma^{-1}$:
\begin{align}
&\phi\mapsto\phi^\Lambda,\qquad \text{such that } \phi^\Lambda(x)=\phi(\Lambda^{-1}x)\nonumber\\
&\Sigma^{-1}\mapsto \Lambda\Sigma^{-1} \Lambda^T,\label{actrotmoins}
\end{align}
where $\Lambda\in O(\gR^D,G)$.
\end{proposition}
\begin{proof}
It is a straightforward computation by using the change of variables $x'=\Lambda^{-1}x$ in the integral, and the identities:
\begin{align*}
&\frac{\partial}{\partial x_\mu}\phi^\Lambda(x)=\frac{\partial\phi}{\partial x_\nu'}(x')\,\Lambda^{-1}_{\nu\mu}, & &\Lambda^{-1} G^{-1}(\Lambda^{-1})^T=G^{-1},\\
&\wx_\mu G^{-1}_{\mu\nu}\wx_\nu=-\frac{4}{\theta^2}x_\mu(G\Sigma^{-1}G\Sigma^{-1}G)_{\mu\nu}x_\nu, & &\Lambda^T G\Lambda=G.
\end{align*}
\end{proof}

Starting from a fixed symplectic structure $\Sigma_0$ in the theory, one can then reach the whole orbit of $\Sigma_0$ with the action of the orthogonal group, namely at the level of the inverses of the symplectic structures:
\begin{equation}
\mathcal O^{(-1)}_{\Sigma_0}=\{\Lambda\Sigma_0^{-1}\Lambda^T,\quad \Lambda\in O(\gR^D,G)\}.
\end{equation}
There exists a particular orbit of this action, composed of the adapted symplectic structures $\Sigma$ for a fixed $G$. This orbit will be crucial to restore the symmetry of the whole orthogonal group at the quantum level.

\subsection{The orbit of adapted symplectic structures}

Let us give some basic definitions useful in the following.

\begin{definition}
Let $V$ be a real vector space of finite even dimension $D$ and $G$ be a positive definite scalar product on $V$.
\begin{itemize}
\item A linear map $I\in\caL(V)$ is called a complex structure on $V$ if $I^2=-\gone$.
\item A symplectic structure $\Sigma$ on $V$ is said to be adapted to $G$ if there exists a complex structure $I$ on $V$ satisfying:
\begin{equation*}
I^TGI=G\text{ and }\Sigma=I^TG,
\end{equation*}
in matrix notations.
\end{itemize}
\end{definition}

\begin{example}
Consider the vector space $\gR^D$, endowed with the standard euclidean metric $G^{\text{st}}$, and the standard symplectic structure $\Sigma^{\text{st}}$, which are given by \eqref{stmetric} in the four-dimensional case. Then, $\Sigma^{\text{st}}$ is adapted to $G^{\text{st}}$ by the complex structure $I^{\text{st}}=-\Sigma^{\text{st}}$.
\end{example}

We recall the standard theorem (see textbooks on the subject, for example \cite{Moroianu:2007}) about adapted symplectic structures and its proof, in order to define some notations which will be useful in the following.
\begin{theorem}
\label{thmorbit}
For a given positive scalar product $G$ of $\gR^D$, the set of all symplectic structures adapted to $G$ is exactly one orbit of the action \eqref{actrot} of the orthogonal group $O(\gR^D,G)$.
\end{theorem}
\begin{proof}
\begin{itemize}
\item Let $Sympl_G(\gR^D)$ be the set of all symplectic structures adapted to $G$. Then, $\forall\Sigma\in Sympl_G(\gR^D)$, ($\Sigma=I^TG$) $\forall \Lambda\in O(\gR^D,G)$, $(\Lambda^{-1})^T\Sigma\Lambda^{-1}$ is a symplectic form adapted to $G$ by the complex structure $\Lambda I\Lambda^{-1}$. $Sympl_G(\gR^D)$ is therefore a certain number of orbits of the action \eqref{actrot}.
\item Since $G$ is positive definite real symmetric, there exists a $G^{\frac 12}$ also real symmetric. Let $\Sigma_0=G^{\frac12}\Sigma^{\text{st}}G^{\frac12}$. $\Sigma_0$ is a symplectic form of $\gR^D$ associated to $G$ by the complex structure $I_0=-G^{-\frac12}\Sigma^{\text{st}}G^{\frac12}=G^{-\frac12}I^{\text{st}}G^{\frac12}$.
\item If $\Sigma\in Sympl_G(\gR^D)$, let $I$ be the complex structure associating $\Sigma$ to $G$. Then,
\begin{equation}
J=-G^{-\frac12}\Sigma G^{-\frac12}=G^{\frac12}IG^{-\frac12}\in O(\gR^D,\gone).
\end{equation}
Indeed, $J^TJ=-G^{\frac12}I^2G^{-\frac12}=\gone$. As a consequence, and since $J$ is also antisymmetric: $J^T=-J$, the matrix $J$ can be written as
\begin{equation}
J=-S^T\begin{pmatrix} \varepsilon_1\sigma & 0 &  & 0 \\ 0 & \varepsilon_2\sigma &  & 0 \\  &  & \ddots &  \\ 0 & 0 &  & \varepsilon_{\frac D2}\sigma \end{pmatrix}S,
\end{equation}
where $\sigma=\begin{pmatrix}0&-1\\1&0\end{pmatrix}$ and $\varepsilon_\alpha\in\{\pm 1\}$. Let $\rho=\begin{pmatrix}0&1\\1&0\end{pmatrix}$. One has: $\rho^T\sigma\rho=-\sigma$ and $\rho\in O(\gR^2,\gone_2)$. Let $S_0$ be the block diagonal $D\times D$ matrix given by: if $\varepsilon_\alpha=1$, the $\alpha^{\text{th}}$ $2\times 2$ diagonal block of $S_0$ is $\gone_2$, else it is $\rho$. Then, with $R=S_0S\in O(\gR^D,\gone)$,
\begin{equation}
J=-R^T(\Sigma^{\text{st}})R.
\end{equation}
One has the decomposition: $\Sigma=G^{\frac12}R^T(\Sigma^{\text{st}})RG^{\frac12}$.
\item By setting $\Lambda^{-1}=G^{-\frac12}RG^{\frac12}$, one has: $\Lambda\in O(\gR^D,G)$, and $\Sigma=(\Lambda^{-1})^T\Sigma_0\Lambda^{-1}$, which shows that $\Sigma$ is in the orbit $\mathcal O_{\Sigma_0}$ of $\Sigma_0$.
\end{itemize}
\end{proof}

For $\Sigma$ adapted to the scalar product $G$, this result permits to exhibit an isomorphism $\caMst\to\caM$ given by $f\mapsto f(RG^{\frac12}\fois)$, in the above notations.

\section{Noncommutative QFT for adapted symplectic structures}

\subsection{Invariance of the quantum action under orthogonal group}

We consider in this subsection the $D$-dimensional Moyal space endowed with a positive scalar product $G$ and a symplectic structure $\Sigma$, adapted to $G$. In this setting, it is possible to prove that the scalar theory with harmonic term \eqref{actscalcurv} is invariant under the orthogonal group at the quantum level.

The propagator of the theory \eqref{actscalcurv} is defined to satisfy to:
\begin{equation}
\left(-\frac 12G^{-1}_{\mu\nu}\partial^x_\mu\partial^x_\nu+\frac{\Omega^2}{2}G^{-1}_{\mu\nu}\wx_\mu\wx_\nu+\frac{m^2}{2} \right)C(x,y)=\delta(x-y)
\end{equation}

\begin{proposition}
Since $\Sigma$ is adapted to $G$, the propagator of the action \eqref{actscalcurv} takes the following form:
\begin{align}
C(x,y)&=\frac{\theta\sqrt{|\det(G)|}}{4\Omega}\left(\frac{\Omega}{\pi\theta}\right)^{\frac D2}\int_0^\infty \!\! \frac{d\alpha}{\sinh^{\frac D2}(\alpha)} e^{-\frac{m^2\alpha}{2\widetilde \Omega}}C(x,y,\alpha),\label{eq-propag}\\
C(x,y,\alpha) &= \exp\Big(-\frac{\widetilde \Omega}{4}\coth(\frac \alpha 2)(x-y)G(x-y) -\frac{\widetilde \Omega}{4}\tanh(\frac \alpha 2)(x+y)G(x+y)\Big)
\end{align}
(where $\widetilde \Omega=\frac{2\Omega}{\theta}$), and is independent of $\Sigma$.
\end{proposition}
\begin{proof}
Indeed, $G_{\mu\nu}^{-1}\wx_\mu\wx_\nu=\frac{4}{\theta^2}G_{\mu\nu}x_\mu x_\nu$, using the property $(\Sigma^{-1})^TG\Sigma^{-1}=G^{-1}$. Then, one can finish the proof like in the paper \cite{Gurau:2005qm}.
%Note that one cannot find a propagator of this form if $\Sigma$ is not adapted to $G$.
\end{proof}

Let us now derive the Feynman amplitudes of such a scalar theory.
\begin{lemma}
\label{lemmainter}
The interaction of the action \eqref{actscalcurv} takes the form:
\begin{align}
\int d^Dx(\phi\star_\Sigma\phi\star_\Sigma\phi\star_\Sigma\phi)(x)&=\frac{\det(G)^2}{\pi^{2D}\theta^{2D}}\int d^Dp \left(\prod_{i=1}^4 d^Dx_i\phi(x_i)\right) V(x_1,x_2,x_3,x_4,p),\\
V(x_1,x_2,x_3,x_4,p)&=\exp\left(-i\sum_{i<j}(-1)^{i+j+1}x_i\wedge x_j-i\sum_{i}(-1)^{i+1}p\wedge x_i\right),
\end{align}
where $p$ is called the hypermomentum associated to the vertex.
\end{lemma}
\begin{proof}
Indeed, the vertex can be expressed as:
\begin{equation}
\int d^Dx(\phi\star_\Sigma\phi\star_\Sigma\phi\star_\Sigma\phi)(x)=\frac{\det(G)}{(\pi\theta)^D}\int \left(\prod_{i=1}^4 d^Dx_i\phi(x_i)\right)\delta(x_1-x_2+x_3-x_4)e^{-i\sum_{i<j}(-1)^{i+j+1}x_i\wedge x_j},\label{interaction}
\end{equation}
since $|\det(\Sigma)|=\det(G)>0$. We get the result thanks to the following identity:
\begin{equation}
\delta(x_1-x_2+x_3-x_4)=\frac{\det(G)}{\pi^D\theta^D}\int d^Dp\ e^{-ip\wedge(x_1-x_2+x_3-x_4)},
\end{equation}
\end{proof}

Consider now a graph $\Gamma$ with a set $V$ of $n$ internal vertices, $N$ external legs and a set $L$ of $(2n-\frac N 2)$ internal lines or propagators. In the following, we will compare the amplitudes of such a graph in the theory with a positive metric $G$ and symplectic form $\Sigma$, and the same graph in the standard theory with $G^{\text{st}}$ and $\Sigma^{\text{st}}$.

In the theory defined by \eqref{actscalcurv}, there are four positions (called ``corners'') associated to each vertex $v\in V$, and each corner is bearing either a half internal line or an external field. These corners are denoted by $x^v_i$, where $i\in\{1,..,4\}$ is given by the cyclic order of the Moyal product. The set $I\subset V\times\{1,..,4\}$ of internal corners (hooked to some internal line) has $4n-N$ elements whereas the set $E=V\times\{1,..,4\}\setminus I$ of external corners has $N$ elements. Each vertex $v\in V$ carries also a hypermomentum, which is noted $p_v$. A line $l\in L$ of the graph $\Gamma$ joins two corners in $I$, and we note their positions by $x^{l,1}$ and $x^{l,2}$. We will also note the external corners $x_e$. Be care of the fact that each corner has two notations for its position.

%\pagebreak

\begin{theorem}
\label{thm-ampl}
In the theory \eqref{actscalcurv},
\begin{enumerate}
\item the regularized amputated amplitudes of Feynman graphs depend on a simple way of the symplectic structure $\Sigma$ and of the scalar product $G$. There exists $R\in O(\gR^D,\gone)$, such that $\Sigma^{-1}=G^{-\frac12}R^T(\Sigma^{\text{st}})^{-1}RG^{-\frac12}$, and
\begin{equation}
\caA_\Gamma^{G,\Sigma^{-1}}(\{x_e\})=(\det G)^{\frac12(n+\frac N2)}\caA_\Gamma^{\text{st}}(\{RG^{\frac12}x_e\}),\label{resultampl}
\end{equation}
where $\caA_\Gamma^{\text{st}}$ is the amputated amplitude for the graph $\Gamma$ corresponding to the theory with standard scalar product $G^{\text{st}}$ and symplectic structure $\Sigma^{\text{st}}$.
\item the regularized amputated amplitudes of Feynman graphs are invariant under the action of the orthogonal group:
\begin{equation}
\caA_\Gamma^{G,\Lambda\Sigma^{-1}\Lambda^T}(\{\Lambda x_e\})=\caA_\Gamma^{G,\Sigma^{-1}}(\{x_e\}),\qquad \forall\Lambda\in O(\gR^D,G).
\end{equation}
\end{enumerate}
\end{theorem}
\begin{proof}
\begin{itemize}
\item We choose a usual regularization for noncommutative field theories on the Moyal space: the $\alpha$ parameter in \eqref{eq-propag} is integrated between $\epsilon$ and $\infty$, where $\epsilon$ is a UV cut-off.
\item Thanks to Lemma \ref{lemmainter}, we can express the amputated amplitude $\caA_\Gamma^{G,\Sigma^{-1}}$ of a graph $\Gamma$ as (for more details, see \cite{deGoursac:2007qi}):
\begin{align}
\caA_\Gamma^{G,\Sigma^{-1}}(\{x_e\})=&\left(\frac{\theta\sqrt{\det(G)}}{4\Omega}\Big(\frac{\Omega}{\pi\theta}\Big)^{\frac D2}\right)^{2n-\frac N 2} \left(\frac{\det(G)}{\pi^D\theta^D}\right)^{2n} \int_\epsilon^\infty \prod_{l\in L}\frac{d\alpha_l}{\sinh^{\frac D2}(\alpha_l)} \int \prod_{(v,i)\in I}d^Dx^v_i\nonumber\\
&\times\prod_{v\in V}d^Dp_v\prod_{l\in L}C(x^{l,1},x^{l,2},\alpha_l)\prod_{v\in V} V(x^v_1,x^v_2,x^v_3,x^v_4,p_v).
\end{align}
\item In the notations of the proof of Theorem \ref{thmorbit}, there exists $R\in O(\gR^D,\gone)$ such that
\begin{equation}
\Sigma^{-1}=G^{-\frac12}R^T(\Sigma^{\text{st}})^{-1}RG^{-\frac12}.
\end{equation}
Let us do the following change of variables for internal corners and hypermomenta:
\begin{equation}
y_i^v=RG^{\frac12}x_i^v,\quad q_v=RG^{\frac12}p_v.
\end{equation}
Since $R^TR=\gone$,
\begin{align}
C(x^{l,1},x^{l,2},\alpha_l)=& \exp\Big(-\frac{\widetilde \Omega}{4}\coth(\frac \alpha 2)(x^{l,1}-x^{l,2})G^{\frac12}R^TRG^{\frac12}(x^{l,1}-x^{l,2})\nonumber\\
& -\frac{\widetilde \Omega}{4}\tanh(\frac \alpha 2)(x^{l,1}+x^{l,2})G^{\frac 12}R^T RG^{\frac 12}(x^{l,1}+x^{l,2})\Big)\nonumber\\
=&\exp\Big(-\frac{\widetilde \Omega}{4}\coth(\frac \alpha 2)(y^{l,1}-y^{l,2})^2 -\frac{\widetilde \Omega}{4}\tanh(\frac \alpha 2)(y^{l,1}+y^{l,2})^2 \Big)\nonumber\\
=&C^{\text{st}}(y^{l,1},y^{l,2},\alpha_l),
\end{align}
and
\begin{align}
V(x_i^v,p_v)=&\exp\Big(-\frac{2i}{\theta}\sum_{i<j}(-1)^{i+j+1}(x_i^v)^TGG^{-\frac12}R^T(\Sigma^{\text{st}})^{-1} RG^{-\frac12}Gx_j^v\nonumber\\
&-\frac{2i}{\theta}\sum_i(-1)^{i+1}p_v^TGG^{-\frac12}R^T(\Sigma^{\text{st}})^{-1}RG^{-\frac12}Gx_i^v\Big)\nonumber\\
=&\exp\Big(-\frac{2i}{\theta}\sum_{i<j}(-1)^{i+j+1}(y_i^v)^T(\Sigma^{\text{st}})^{-1}y_j^v -\frac{2i}{\theta}\sum_i(-1)^{i+1}q_v^T(\Sigma^{\text{st}})^{-1}y_i^v\Big)\nonumber\\
=&V^{\text{st}}(y_i^v,q_v).
\end{align}
The global Jacobian of this change of variables is $|(\det(RG^{\frac12})^{-1})^{2(2n-\frac N2)+n}|=\det(G)^{-\frac{5n}{2}+\frac N2}$. With the other factors involving $\det(G)$: $\det(G)^{\frac 12(2n-\frac N2)+2n}$, one finds:
\begin{equation}
\caA_\Gamma^{G,\Sigma^{-1}}(\{x_e\})=(\det G)^{\frac12(n+\frac N2)}\caA_\Gamma^{\text{st}}(\{RG^{\frac12}x_e\}).
\end{equation}
Since the changes of variables do not affect the $\alpha_l$-variables, the latter result is true for regularized amplitudes of Feynman graphs.
\item The action of $\Lambda\in O(\gR^D,G)$ on $\caA_\Gamma$ is:
\begin{equation}
(\caA_\Gamma^{G,\Sigma^{-1}}(\{x_e\}))^\Lambda=\caA_\Gamma^{G,\Lambda\Sigma^{-1}\Lambda^T}(\{\Lambda x_e\}).
\end{equation}
Then,
\begin{equation}
\Lambda\Sigma^{-1}\Lambda^T=\Lambda G^{-\frac12}R^T(\Sigma^{\text{st}})^{-1} RG^{-\frac12}\Lambda^T=G^{-\frac12}R_\Lambda^T(\Sigma^{\text{st}})^{-1} R_\Lambda G^{-\frac12},
\end{equation}
if one sets $R_\Lambda=RG^{-\frac12}\Lambda^TG^{\frac12}$. Since $\Lambda G^{-1}\Lambda^T=G^{-1}$, $R_\Lambda\in O(\gR^D,\gone)$:
\begin{equation}
R_\Lambda^T R_\Lambda=G^{\frac12}\Lambda G^{-\frac12}R^TRG^{-\frac12}\Lambda^TG^{\frac12}=\gone.
\end{equation}
By using the first point of the proof,
\begin{equation}
\caA_\Gamma^{G,\Lambda\Sigma^{-1}\Lambda^T}(\{\Lambda x_e\})=(\det G)^{\frac12(n+\frac N2)} \caA_\Gamma^{\text{st}}(\{R_\Lambda G^{\frac12}\Lambda x_e\}).
\end{equation}
Since $R_\Lambda G^{\frac12}\Lambda=RG^{\frac12}$, one finds
\begin{equation}
\caA_\Gamma^{G,\Lambda\Sigma^{-1}\Lambda^T}(\{\Lambda x_e\})=(\det G)^{\frac12(n+\frac N2)} \caA_\Gamma^{\text{st}}(\{RG^{\frac12}x_e\})=\caA_\Gamma^{G,\Sigma^{-1}}(\{x_e\}).
\end{equation}
\end{itemize}
\end{proof}

\begin{corollary}
\label{cor-effact}
The bare (regularized) effective action of the theory \eqref{actscalcurv} is invariant under the orthogonal group.
\end{corollary}
\begin{proof}
The effective action is defined to be:
\begin{equation}
\caG^{G,\Sigma^{-1}}(\phi)=\sum_{\Gamma}\alpha(\Gamma)\int \dd x_1\dots\dd x_{N(\Gamma)} \caA_\Gamma^{G,\Sigma^{-1}}(\{x_i\}) \phi(x_1)\dots\phi(x_{N(\Gamma)}),
\end{equation}
where the sum on $\Gamma$ is over the $1PI$ Feynman graphs, $\alpha(\Gamma)$ is a factor depending uniquely on the topology of the graph $\Gamma$, and $N(\Gamma)$ is the number of the external legs of $\Gamma$. Then:
\begin{equation*}
\caG^{G,\Lambda\Sigma^{-1}\Lambda^T}(\phi^\Lambda)=\sum_{\Gamma}\alpha(\Gamma)\int \dd x_1\dots\dd x_{N(\Gamma)} \caA_\Gamma^{G,\Lambda\Sigma^{-1}\Lambda^T}(\{x_i\}) \phi(\Lambda^{-1}x_1)\dots\phi(\Lambda^{-1}x_{N(\Gamma)}).
\end{equation*}
By the change of variables $x'_i=\Lambda^{-1}x_i$, one obtains:
\begin{align*}
\caG^{G,\Lambda\Sigma^{-1}\Lambda^T}(\phi^\Lambda)&=\sum_{\Gamma}\alpha(\Gamma)\int \dd x'_1\dots\dd x'_{N(\Gamma)} \caA_\Gamma^{G,\Lambda\Sigma^{-1}\Lambda^T}(\{\Lambda x'_i\}) \phi(x'_1)\dots\phi(x'_{N(\Gamma)})\\
&=\caG^{G,\Sigma^{-1}}(\phi),
\end{align*}
thanks to Theorem \ref{thm-ampl}.
\end{proof}

Since the orthogonal group acts also on the symplectic structure $\Sigma$, one says sometimes that the effective action is covariant under this group.

\subsection{Renormalizability of QFT}

As a consequence of Theorem \ref{thm-ampl}, one obtains also the renormalizability of scalar theories with harmonic term for a general positive scalar product $G$ and a symplectic structure $\Sigma$ adapted to $G$.

\begin{corollary}
In the four-dimensional case, the theory \eqref{actscalcurv} is renormalizable up to all orders in perturbation, for all positive scalar products $G$, and symplectic structures $\Sigma$ adapted to $G$.
\end{corollary}
\begin{proof}
The regularized amplitudes $\caA_\Gamma^{G,\Sigma^{-1}}$ can be expressed in a simple way in terms of $\caA_\Gamma^{\text{st}}$ (see Equation \eqref{resultampl}), they have therefore the same divergences. Since the scalar theory with harmonic term for $G^{\text{st}}$ and $\Sigma^{\text{st}}$ is renormalizable at all orders \cite{Grosse:2004yu} for $D=4$, the same renormalization procedure can be applied for the theory \eqref{actscalcurv}. Then, this theory is renormalizable up to all orders in perturbation, and moreover the relation \eqref{resultampl} is still verified for the renormalized amplitudes $\caA_\Gamma^{G,\Sigma^{-1}}$.
\end{proof}

Since the renormalized amplitudes satisfy Equation \eqref{resultampl}, by using the same arguments as in the proof of Theorem \ref{thm-ampl} and Corollary \ref{cor-effact}, one can deduce that:
\begin{corollary}
The renormalized effective action of the renormalizable theory \eqref{actscalcurv} is invariant under the action of the orthogonal group.
\end{corollary}

\subsection{Conclusion}

We have considered in this paper a scalar field theory with harmonic term on the Moyal algebra $\caM$, with a general (positive) scalar product $G$ and a symplectic structure $\Sigma$, presented in subsection 2.1. For a given $\Sigma$, the classical action \eqref{actscalcurv} of this theory is not invariant under the whole orthogonal group. However, if one allows that the orthogonal group acts also on $\Sigma$ by \eqref{actrot}, one recovers the invariance of the classical action under this group, as exposed in subsection 2.2. But, $\Sigma$ is now moving in the orbits of the action \eqref{actrot}. In subsections 2.3, we have seen that there exists a particular orbit of this action: the orbit of symplectic structures adapted to a fixed scalar product $G$, namely symplectic structures related to $G$ by a complex structure.

If we suppose that $\Sigma$ belongs to this special orbit, one obtains  several properties at the quantum level: the propagator of the theory has an analogous expression as in the standard case, the regularized amputated amplitudes of Feynman graphs can be simply reexpressed in terms of standard amplitudes. As a consequence, we have seen in subsection 3.1 that the bare effective action is invariant under the whole orthogonal group. In subsection 3.2, we have obtained the renormalizability of the theory \eqref{actscalcurv} in the four-dimensional case, provided that $\Sigma$ is adapted to $G$, and the renormalized effective action is also invariant under the action of the orthogonal group. The theory \eqref{actscalcurv} admits therefore all the rotations as symmetries, even if $D>2$.

Note that the investigation of all the vacuum configurations of the Grosse-Wulkenhaar model by assuming their invariance under the group $SO(D)\cap Sp(D)$ gives in fact solutions which are invariant under the whole group $O(D)$ \cite{deGoursac:2007uv,deGoursac:2009gh}. This can be explained by the hidden symmetry of the theory under $O(D)$ as shown above.

The use of adapted symplectic structures has been a key ingredient in all the proofs, so that the restoration of the invariance under the whole orthogonal group at the quantum level, as it is described here, is deeply related to the particular orbit of adapted symplectic structures. Furthermore, these results give also a satisfying answer to the problem of naturality of a symplectic structure on a position space. Indeed, the theory \eqref{actscalcurv} does not depend on a fixed symplectic structure anymore, but only from a special orbit of the action \eqref{actrot}, uniquely determined by the scalar product $G$.

These results cannot be directly adapted to the Minkowskian framework, namely for a scalar product $G$ of signature $(D-1,1)$, since the set of adapted symplectic structures (by complex structures) to such a scalar product $G$ is empty. However, it would be interesting to study if the results of this paper can be extended to a curved framework, for instance a symmetric space on which the Langmann-Szabo duality has been implemented \cite{Bieliavsky:2008qy}.

\vskip 1 true cm
\noindent
{\bf{Acknowledgements}}: One of the authors (A.G.) thanks Raimar Wulkenhaar for interesting discussions on this work.

\bibliographystyle{utcaps}
\bibliography{biblio-these,biblio-perso,biblio-recents}

\end{document}